\documentclass[leqno,12pt]{article}

\usepackage{mathrsfs}  
\usepackage{amsmath}
\usepackage{amssymb}          
\usepackage{graphicx}               
\usepackage{latexsym}              
\usepackage{amsfonts}  
\usepackage{amsthm}

\newtheorem{theorem}{Theorem}      
\newtheorem{lemma}{Lemma}     
\newtheorem{definition}{Definition} 
\newtheorem{corollary}{Corollary}

\newtheorem{proposition}{Proposition}
\DeclareMathOperator{\sech}{sech}

\newcommand{\f}{\frac}

\numberwithin{equation}{section}
\numberwithin{theorem}{section}
\numberwithin{lemma}{section}
\numberwithin{definition}{section}
\numberwithin{corollary}{section}

\begin{document}
\centerline{{\bf \large Maximal Densities}}
\vspace{.1in}
\centerline{{\bf \large of Finite-Gap Solutions}}
\vspace{.1in}
\centerline{{\bf \large of the Sine-Gordon Equation}}

\vspace{.2cm}
\centerline{{Otis C. Wright, III}}
\centerline{{\em School of Science and Mathematics}}
\centerline{{\em Cedarville University}}
\centerline{{\em 251 N. Main St.}}
\centerline{{\em Cedarville, Ohio 45314}}
\centerline{E-mail: wrighto@cedarville.edu}
\centerline{{\today}}

\vspace{.2cm}

\centerline{{\bf \large Abstract}}
We establish a  sharp upper bound on the densities of finite-gap solutions of the sine-Gordon equation. The bound is derived directly from the finite-dimensional hierarchy, without explicit integration of the finite-gap solutions. The maximal density is determined by the roots of the invariant polynomial. An analogous sharp upper bound is established for a bounded class of finite-gap solutions of the sinh-Gordon equation. 


\section{Introduction}
\label{introduction}

The sine-Gordon equation 
\begin{equation}
\phi_{xt} =  \sin \phi,
\label{sG}
\end{equation}
where $\phi = \phi (x,t)$ is a real-valued wave profile, is one of the most important completely integrable nonlinear wave equations in mathematical physics. A closely related equation is the sinh-Gordon  equation,
\begin{equation}
\phi_{xt} = \sinh \phi.
\label{shG}
\end{equation}
Both equations arise in  differential geometry, relativistic field theory, nonlinear optics, Josephson junctions, and the theory of surfaces of constant curvature~\cite{cost 78, ablo 11, cuev 14}. Like the Korteweg-de Vries (KdV), modified Korteweg-de Vries (mKdV), and nonlinear Schr\"odinger (NLS) equations, they possess infinitely many conservation laws, admit soliton solutions, and can be solved by inverse-scattering and algebro-geometric methods~\cite{akns 73a, akns 73b, akns 74, hiro 72, ablo 81, date 82, dubr 82, belo 94}. Also, the sine-Gordon equation is  part of a combined sG/mKdV hierarchy, placing both equations within a common integrable framework~\cite{gesz 00, gesz 03}.

The finite-dimensional construction used in this paper originates in the
commuting-ODE framework developed by Its, Kotlyarov, and Kozel~\cite{its 76, itsk 76, kotl 76, koze 76}. In this approach, finite-gap solutions are constructed from
commuting systems of ordinary differential equations generated by a
matrix-polynomial ansatz. 
Recently, this framework was used to  derive
sharp amplitude bounds for finite-gap solutions of several nonlinear
Schr\"odinger-type equations and for the modified Korteweg-de Vries
equation~\cite{wrig 19, wrig 20, wrig 24, wrig 26}. Bertola and Tovbis obtained sharp maximal-amplitude formulas for the focusing and defocusing nonlinear Schr\"odinger equations using algebro-geometric and Riemann-Hilbert methods~\cite{bert 17}.
A similar finite-dimensional framework was used by Bernatska to construct finite-gap solutions of the mKdV and sine-Gordon equations~\cite{bern 25, bern 26}. 
Related finite-dimensional descriptions
of finite-gap solutions of the sine-Gordon equation were developed by
Ercolani, Forest, Flaschka, and McLaughlin through polynomial recursion operators and
isospectral flows on finite-dimensional invariant manifolds~
\cite{fore 82, erco 85, erco 90}.

The finite-gap
solutions of the sine-Gordon and sinh-Gordon equations arise naturally on 
the same finite-dimensional phase space previously used to analyze maximal amplitudes of
finite-gap solutions of the mKdV equation~\cite{wrig 26}. The same shifted polynomial ansatz associated with the principal grading of the loop algebra is used to construct the hierarchy of commuting flows, instead of the more common AKNS formulation. Within this framework, the hierarchy reduction and the optimization mechanism are naturally implemented. The same finite-dimensional optimization argument applies to both the mKdV and the sine-Gordon equations, illustrating that it is a property of the finite-dimensional hierarchy rather than of the individual reductions. The final form of the sharp upper bound formula is conveniently expressed in terms of parameters related by the gauge equivalence of the principal-grading and AKNS formulations.  

This paper establishes a sharp upper bound for the densities of finite-gap solutions of the  sine-Gordon equation. The  bound  is derived without explicit integration of the finite-gap solutions.  Instead, the critical-point equations for the maximizing configuration induce a natural
factorization of the invariant polynomial that expresses the physically relevant variable in terms
of the spectral invariants.

In particular, we show that for  finite-gap (\(N\)-phase) solutions of the sine-Gordon equation, the maximizing configuration induces a natural factorization of the invariant polynomial \(\mathscr{R}(\lambda)\) on the double cover \(\lambda = -E^2\). The parameter $\lambda$ is the spectral parameter of the principal grading, and the parameter $E$ is  the AKNS spectral parameter under the gauge equivalence established in Appendix A. Let
$\mathcal E^+$ be the set of corresponding  upper-half-plane spectral points, viz.,
\begin{equation}
\mathcal E^+ = \{ E \in \mathbb{C}: \Im \left( E \right) >0, \mathscr{R} \left( - E^2 \right) = 0 \}.
\end{equation}
Then the density $\phi_x,$ corresponding to the local kink density of the classical soliton, satisfies the sharp upper bound
\begin{equation}\label{sharp}
|\phi_{x} (x,t)|
\le
2 \sum_{E \in\mathcal E^+}\Im \left(E\right).
\end{equation}
The factor $2$ in the sharp upper bound formula~(\ref{sharp}) arises from the normalization of the sine-Gordon equation~(\ref{sG}).
The bound is sharp in the sense that there exists a solution that attains the upper bound. An analogous formula is established for a bounded class of solutions of the sinh-Gordon equation~(\ref{shG}).

The remainder of this paper is organized as follows. Section 2 establishes a concrete formulation of the finite-gap sine-Gordon solutions using dynamical variables in a finite-dimensional hierarchy of commuting systems of ordinary differential equations derived from a  loop algebra of $\mathfrak{sl}(2,\mathbb{C})$ with a shifted polynomial ansatz. Section 3 details the direct analysis at the level of the dynamical variables of the critical points of these commuting flows, culminating in the derivation of the sharp upper bound formula for the density of a finite-gap solution of the  sine-Gordon equation~(\ref{sG}) and an analogous formula for the sinh-Gordon equation~(\ref{shG}). Appendix A discusses the relation to the AKNS spectral problem, while Appendix B gives genus-one examples of the sharp upper bound formulas.

\section{Finite-Dimensional Hierarchy}

In this work the sine-Gordon equation~(\ref{sG}) is studied within a  loop algebra of $\mathfrak{sl} (2,\mathbb{C})$. The density $\phi_x (x,t)$ is associated directly with the diagonal entry of a generating matrix Laurent polynomial $\Psi^{(N)}.$ The Laurent polynomials used to define the commuting hierarchy flows satisfy a shifted ansatz with respect to the spectral parameter $\lambda,$ and a nonhomogeneous splitting of the loop algebra ensures that the flows are well defined and commute. This construction corresponds to the principal grading of the loop algebra~\cite{drin 85, kac 90}. The resulting spectral problem is gauge equivalent to the AKNS formulation, with principal-grading spectral parameter $\lambda$ and the AKNS spectral parameter $E$ related by $\lambda=-E^2$; see Appendix A.

The finite-gap solutions are defined as solutions of $N$ nonlinear autonomous ordinary differential equations generated by $3N$ dependent complex dynamical variables subject to invariant reality constraints, yielding an effective $3N$-dimensional real phase space~\cite{wrig 26, bern 25, bern 26}. These $N$ ordinary differential equations commute with each other and possess a common set of invariants, allowing the compatible local solutions to extend globally. 
The sine-Gordon equation is commonly interpreted  as a negative flow in the integrable hierarchy. A distinctive feature of the finite-dimensional  projection is that the first negative hierarchy flow is represented, up to sign, by the $(N-1)$th positive flow of the finite hierarchy. Consequently, the sine-Gordon equation is realized within the same finite-dimensional hierarchy of commuting positive flows used previously for the mKdV equation.

Although the finite hierarchy of commuting flows remains the same as for the mKdV equation, the reality conditions and the time flow are different for the sine-Gordon equation. The different reductions lead to distinct geometries of the spectral invariants, yet the optimization mechanism persists and yields closely related sharp upper bound formulas.

We now recall the important features of the construction of the combined sG/mKdV finite-dimensional hierarchy~\cite{wrig 26} and introduce the reality constraints and normalization condition that realize the sine-Gordon time flow within the $(N-1)$th flow of the hierarchy.

\begin{definition}[Loop Algebra and Splitting Operators]
Let $\mathfrak{g} = \mathfrak{sl}(2,\mathbb{C}) \otimes \mathbb{C} [\lambda, \lambda^{-1}],$ the loop algebra of Laurent polynomials in the spectral parameter $\lambda,$ and let $\Phi \in \mathfrak{g},$ viz.,
\begin{equation}
\Phi = \sum\limits^{M_{+}}_{j=M_{-}} A_j \lambda^j,
\end{equation}
where $A_j \in \mathfrak{sl}(2,\mathbb{C}),$ for $j \in \mathbb{Z}.$ 
Write 
\begin{equation}
A_j = \left(\begin{array}{cc} - a_j & b_j \\ c_j & a_j \end{array}\right),
\end{equation}
for $j=M_{-}, \ldots, M_{+}.$
For any $\Phi \in \mathfrak{g},$ define the splitting operators
\begin{equation}
(\Phi)_+ = \sum\limits^{M_+}_{j=0} A_j \lambda^j - \left(\begin{array}{cc} 0 & 0 \\ c_0 & 0 \end{array} \right)
\end{equation}
and
\begin{equation}
(\Phi)_- = \sum\limits^{-1}_{j=M_{-}} A_j \lambda^j  + \left(\begin{array}{cc} 0 & 0 \\ c_0 & 0 \end{array} \right).
\end{equation}

\end{definition}

The sine-Gordon reality conditions are associated with a different involution of the loop algebra than the mKdV reality conditions.

\begin{lemma}[Commutation of Involution and the Splitting] \label{involutionlemma}
Define the involution $\tau: \mathfrak g \rightarrow \mathfrak g$ by
\begin{equation}
\left( \tau \Phi \right) (\lambda) = J(\lambda) \Phi (\lambda)^T J^{-1} (\lambda),
\end{equation}
where $(\cdot)^T$ is the matrix transpose and 
\begin{equation}
J(\lambda) = \left(\begin{array}{cc} -1 & 0 \\ 0 & \lambda \end{array} \right).
\end{equation}
Then the involution $\tau$ commutes with the projection operators, viz., 
\begin{equation}
\tau \left( \Phi \right)_{\pm}  = \left(   \tau \Phi  \right)_{\pm}. 
\end{equation}
\end{lemma}

\begin{proof}
For any $\Phi \in \mathfrak{g},$ let $(\cdot)_{0}$ be the usual unshifted splitting operator,
\begin{equation}
\left( \Phi \right)_{0} = \sum\limits^{M_+}_{j=0} A_j \lambda^j.
\end{equation}
Then $( \cdot )_{0}$ commutes with the transpose and, also, we can write
\begin{equation}
\left( \Phi \right)_{+} = \left( \Phi \right)_{0} - \left( \begin{array}{cc} 0 & 0 \\ c_0 & 0 \end{array} \right).
\end{equation}
Then
\begin{equation}
\begin{array}{rcl}
 \tau \left( \Phi \right)_{+}  & = &   \tau  \left( \Phi \right)_0 - \tau \left( \begin{array}{cc} 0 & 0 \\ c_0 & 0 \end{array} \right) \\
& = & J \left( \Phi^T \right)_0 J^{-1} - \left(\begin{array}{cc} 0 & -c_0 \lambda^{-1} \\ 0 & 0 \end{array} \right) \\
& = & \left( J \Phi^T J^{-1} \right)_{0} + \left(\begin{array}{cc} 0 & -c_0 \lambda^{-1} \\  b_{-1} & 0 \end{array} \right) - \left(\begin{array}{cc} 0 & -c_0 \lambda^{-1} \\ 0 & 0 \end{array} \right) \\
& =&  \left( J \Phi^T J^{-1} \right)_{0} - \left(\begin{array}{cc} 0 & 0  \\ - b_{-1} & 0 \end{array} \right) \\
& = & \left( J \Phi^T J^{-1} \right)_{+} \\
& = & \left( \tau \Phi  \right)_{+},
\end{array}
\end{equation}
since $-b_{-1}$ is the lower-left coefficient of degree zero in the matrix $J \Phi^T J^{-1}.$
Consequently, 
\begin{equation}
\tau \left(\Phi \right)_{-}  = \tau  \Phi - \tau \left( \Phi \right)_+ = \tau \Phi  - \left( \tau \Phi \right)_+ =  \left( \tau \Phi  \right)_{-}.
\end{equation}
\end{proof}

\begin{corollary}[Commutation of the Derived Involution and the Splitting]
Define the involution $\rho: \mathfrak g \rightarrow \mathfrak g$ by
\begin{equation}
(\rho\Phi)(\lambda) =-\overline{(\tau\Phi)(\overline{\lambda})}.
\end{equation}
Then the involution $\rho$ commutes with the projection operators, viz.,
\begin{equation}
\rho (\Phi)_\pm = \bigl(\rho \Phi \bigr)_\pm.
\end{equation}
\end{corollary}

\begin{proof}
By Lemma~\ref{involutionlemma}
\begin{equation}
\tau (\Phi)_\pm
=
\bigl(\tau \Phi \bigr)_\pm .
\end{equation}
Since complex conjugation commutes with the projection operators,
\begin{equation}
-\overline{
\left(\tau \left( \Phi \right)_{\pm} \right) (\overline{\lambda}) 
}
=
-\overline{
\bigl( \left( \tau \Phi \right)_\pm (\overline{\lambda} )\bigr)
}
=
\left(
-\overline{\left( \tau \Phi \right) (\overline{\lambda})}
\right)_\pm .
\end{equation}
Hence
\begin{equation}
\rho (\Phi)_\pm 
=
\bigl(\rho \Phi \bigr)_\pm .
\end{equation}
\end{proof}

\begin{definition}[Dynamical Variables] \label{dynamicalvariables}
For $N \in \mathbb{N},$ consider $3N$ complex-valued variables $f_j,g_j, h_j \in \mathbb{C},$ for $j=0, \ldots, N-1.$ Define three polynomials
\begin{equation}
F_{N-1}(\lambda) =  \sum\limits_{j=0}^{N-1} f_j \lambda^j
\end{equation}
\begin{equation}
G_{N}(\lambda) = \sum\limits_{j=0}^{N} g_j \lambda^j,
\end{equation}
\begin{equation}
H_{N} (\lambda) = \sum\limits_{j=0}^{N} h_j \lambda^j,
\end{equation}
where $g_{N} = \overline{h_N} =  b \in \mathbb{C}$ is a constant.
Then define the generating matrix
\begin{equation}
\Psi^{(N)} (\lambda) = \left(\begin{array}{cc} -F_{N-1}(\lambda) & G_{N}(\lambda) \lambda^{-1} \\ H_{N} (\lambda) & F_{N-1}(\lambda) \end{array} \right),
\end{equation}
where  $\Psi^{(N)} (\lambda) \in \mathfrak{g} = \mathfrak{sl}(2,\mathbb{C}) \otimes \mathbb{C}[\lambda, \lambda^{-1}].$
\end{definition}

The generating matrix shifts the degrees of the off-diagonal elements relative to each other. The relative shift of the off-diagonal elements is crucial to the construction of  the combined sG/mKdV hierarchy. The nonhomogeneous projection operators respect this relative shift,  splitting the loop algebra into two subalgebras  in the principal grading.

\begin{definition}[Time Flow Matrices]
Define the hierarchy flow matrices
\begin{equation}
V^{(k)}  = ( \lambda^{k-N+1} \Psi^{(N)} )_+,
\end{equation}
for $k=0, \ldots, N-1.$  
\end{definition}

Suppose, also, that  the first negative time flow is defined to be 
\begin{equation}
V^{(-1)} = \left( \Psi^{(N)} \right)_{-},
\end{equation}
 then 
\begin{equation}
\begin{array}{rcl}
\left[ V^{(-1)}, \Psi^{N} \right] & = & \left[ \Psi^{(N)} - \left( \Psi^{(N)} \right)_{+} , \Psi^{(N)} \right] \\[.1in]
&  =&  - \left[ \left( \Psi^{(N)} \right)_{+}, \Psi^{(N)} \right] \\[.1in]
&  = & - \left[V^{(N-1)}, \Psi^{(N)} \right].
\end{array}
\end{equation}
Therefore, the vector field generated by the first negative hierarchy flow coincides, up to sign, with the vector field generated by the 
$(N-1)$th flow of the finite-dimensional hierarchy.  In the following construction, the first negative time flow is never used, instead the sine(sinh)-Gordon flows are constructed  using only the positive time flows.

\begin{definition}[Nonlinear Autonomous Systems]
Define $N$ nonlinear autonomous ordinary differential equations with $N$ independent time variables $t_k \in \mathbb{R},$ for $k=0, \ldots, N-1,$ and $3N$ dependent complex variables
given by the components of $\mathbf{y} \in \mathbb{C}^{3N},$ viz., $y_j=f_{j-1}, y_{N+j}= g_{j-1}, y_{2N+j}= h_{j-1},$ for $j=1, \ldots, N,$  by
\begin{equation}
\frac{\partial}{\partial t_k} \mathbf{y} = \mathbf{F} (\mathbf{y}),
\label{Fode}
\end{equation}
and the components of $\mathbf{F} : \mathbb{C}^{3N} \rightarrow \mathbb{C}^{3N}$ are  polynomials  in the $3N$ dependent variables, $\mathbf{y} \in \mathbb{C}^{3N},$ determined by equating the coefficients of powers of $\lambda$ in the equations
\begin{equation}
\f{\partial \Psi^{(N)}}{\partial t_k} = \left[V^{(k)},\Psi^{(N)} \right],
\label{ode}
\end{equation}
where 
\begin{equation}
[A,B] = AB - BA
\end{equation}
is the usual commutator operator on arbitrary matrices $A, B \in \mathfrak{sl}(2,\mathbb{C}).$
\end{definition}

The hierarchy equations in~(\ref{ode}) define a closed finite-dimensional dynamical system on the Laurent-polynomial ansatz \(\Psi^{(N)}\). The relative shift in the off-diagonal elements and the projection operator act together to close the recursion in powers of $\lambda.$
For details of the closure of the hierarchy equation~(\ref{ode}) and proofs of the following two lemmas, see~\cite{wrig 26}.

\begin{lemma}[Zero-Curvature Identity~\cite{dick 03}]
The $N$ matrix polynomials $V^{(k)},$ for $k = 0, \ldots, N-1,$ satisfy the identity
\begin{equation}
\frac{\partial }{\partial t_\ell} V^{(k)} - \frac{\partial }{\partial t_k} V^{(\ell)}= [V^{(\ell)}, V^{(k)}],
\label{Videntity}
\end{equation}
for $k, \ell = 0, \ldots, N-1.$
\end{lemma}

\begin{lemma}[Commuting Differential Equations~\cite{dick 03}]
The $N$ nonlinear autonomous ordinary differential equations defined by equation~(\ref{ode}) commute with each other, viz.,
\begin{equation}
\frac{\partial}{\partial t_k} \left(\frac{\partial}{\partial t_\ell} \Psi^{(N)} \right) = \frac{\partial}{\partial t_\ell} \left(\frac{\partial}{\partial t_k} \Psi^{(N)} \right),
\end{equation}
for $k, \ell = 0, 1, \ldots, N-1.$
\end{lemma}

Thus the hierarchy generated by $\Psi^{(N)}$ defines compatible finite-dimensional dynamical systems of classical commuting flows. In order to analyze the solutions of the 
sine(sinh)-Gordon equations, we define invariant reality conditions on the dynamical variables. 

\begin{theorem}[Reality Conditions]
Suppose the initial data satisfy either one of the following reality conditions:
\begin{enumerate}
\item Sine-Gordon reality conditions,
\begin{equation}\label{realitysG}
f_j=-\overline{f_j},
\qquad
g_j=\overline{h_j},
\end{equation} for $j=0, \ldots, N,$ with $g_N=\overline{h_N} = b=i.$
\item Sinh-Gordon reality conditions,
\begin{equation}\label{realityshG}
f_j = \overline{f_j}, 
\qquad
g_j =\overline{g_j},
\qquad
h_j=\overline{h_j},
\end{equation} for $j=0, \ldots, N,$ with $g_N=h_N=b=1.$
\end{enumerate}
Then the hierarchy flows preserve the reality conditions.
\end{theorem}

\begin{proof}
The sine-Gordon reality conditions~(\ref{realitysG}) are equivalent to 
\begin{equation}
\left( \rho  \Psi^{(N)} \right) (\lambda)  = - J(\lambda)  \left(\Psi^{(N)} (\overline{\lambda}) \right)^\dagger J^{-1}(\lambda) = \Psi^{(N)} (\lambda).
\end{equation}
where $(\cdot)^\dagger$ indicates the conjugate transpose. Now, for $k=0, 1, \ldots, N-1,$
\begin{equation}
\begin{array}{rcl}
\frac{\partial}{\partial t_k} \left(  \rho \Psi^{(N)} \right)     
& = & - J  \left(\left[V^{(k)} (\overline{\lambda}),\Psi^{(N)} (\overline{\lambda}) \right]\right)^\dagger J^{-1}  \\[.1in]
& = & \left[ -J \left( V^{(k)}(\overline{\lambda}) \right)^\dagger J^{-1}, - J \left( \Psi^{(N)} (\overline{\lambda}) \right)^\dagger J^{-1} \right] \\[.1in]
& = & \left[ \rho V^{(k)} ,  \rho \Psi^{(N)}  \right] \\[.1in]
& = &  \left[   \rho \left( \lambda^{k-N+1} \Psi^{(N)}  \right)_{+} ,   \rho \Psi^{(N)} \right] \\[.2in]
& = &  \left[  \left( \lambda^{k-N+1}  \rho    \Psi^{(N)} \right)_{+}  ,  \rho \Psi^{(N)}  \right]. \\[.2in]
\end{array}
\end{equation}
Since the transformed solution satisfies the same hierarchy equations and agrees with the original solution at the initial time whenever the sine-Gordon reality conditions hold, uniqueness implies that the sine-Gordon reality conditions are preserved for all hierarchy times.

The sinh-Gordon reality conditions~(\ref{realityshG}) are equivalent to
\begin{equation}
\overline{ \Psi^{(N)} (\overline{\lambda} ) } = \Psi^{(N)}(\lambda).
\end{equation}
Now
\begin{equation}
\begin{array}{rcl}
\frac{\partial}{\partial t_k}  \left(  \overline{\Psi^{(N)} (\overline{\lambda})} \right)
& = & \left[\overline{V^{(k)} (\overline{\lambda})},\overline{\Psi^{(N)} (\overline{\lambda})} \right] \\
& = & \left[ \left( \lambda^{k-N+1} \overline{\Psi^{(N)} (\overline{\lambda} ) } \right)_{+}, \overline{\Psi^{(N)} (\overline{\lambda} )} \right].
\end{array}
\end{equation}
As before, the transformed solution satisfies the same hierarchy equations. If the solution agrees with the original solution at the initial time when the sinh-Gordon reality conditions hold, then uniqueness implies that the sinh-Gordon reality conditions are preserved for all hierarchy times.

\end{proof}

\begin{theorem}[Sine(sinh)-Gordon Equations and Hierarchy Relations]\label{equationtheorem}
Let
$
\Psi^{(N)}
$
satisfy the hierarchy equations
\begin{equation} \label{hier}
\frac{\partial}{\partial t_k} \Psi^{(N)}
=
[V^{(k)},\Psi^{(N)}],
\qquad
k=0,\ldots,N-1,
\end{equation}
subject to one of the two  invariant reality conditions~(\ref{realitysG}) or~(\ref{realityshG}).

Then the highest-order recursion relations are
\begin{subequations}\label{top}
\begin{align}
\frac{\partial}{\partial t_k} g_{N} &= 0,\label{topA} \\
\frac{\partial}{\partial t_k} h_{N} &=0,\label{topB}
\end{align}
\end{subequations}
and
\begin{subequations}\label{leadingorder}
\begin{align}
\frac{\partial}{\partial t_k} f_{N-1}
&=
\overline{b} g_{N-k-1}-b h_{N-k-1}, \label{leadingA}\\
\frac{\partial}{\partial t_k} g_{N-1}
&=
2 b f_{N-k-2}-2 f_{N-1} g_{N-k-1}, \label{leadingB}\\
\frac{\partial}{\partial t_k} h_{N-1}
&=
-2 \overline{b} f_{N-k-2} + 2 f_{N-1} h_{N-k-1},\label{leadingC}
\end{align}
\end{subequations}
for \(k=0,\ldots,N-1\), with the convention that $f_j=g_j=h_j=0,$ for $j <0,$ and $g_{N} = \overline{h_N} = b \in \mathbb{C}.$
And the lowest-order recursion relations are
\begin{subequations} \label{lowest}
\begin{align}
\frac{\partial}{\partial t_{k}} f_0 & = h_{N-k} g_0 - g_{N-k} h_0, \label{lowestA}\\
\frac{\partial}{\partial t_{k}} g_0 &=  -2 f_{N-k-1} g_0, \label{lowestB}\\
\frac{\partial}{\partial t_{k}} h_0  &=  2 f_{N-k-1} h_0, \label{lowestC}
\end{align}
\end{subequations}
for $k=0, \ldots, N-1.$

With sine-Gordon reality conditions~(\ref{realitysG}), restrict to the invariant manifold determined by the normalization condition
\begin{equation}
|g_0|^2=|h_0|^2=\frac{1}{16},
\end{equation}
and define the real-valued function $\phi$ by $g_0 = -\frac{i}{4} e^{-i \phi}.$ 

With sinh-Gordon reality conditions~(\ref{realityshG}), restrict to the invariant manifold determined by the normalization condition
\begin{equation}
g_0 h_0 = \frac{1}{16},
\end{equation}
with $g_0>0,$ and define the real-valued function $\phi$ by  $g_0 =\frac{1}{4}e^{\phi}.$  Then,  when $N=1,$ $\phi$ generates a  traveling-wave solution of the sine(sinh)-Gordon equations~(\ref{sG}, \ref{shG}) in the $k=0$ flow. When $N \geq 2,$ $\phi$  generates a solution to the sine(sinh)-Gordon equations~(\ref{sG}, \ref{shG}) in the $k=0$ and $k=N-1$ flows.
\end{theorem}

\begin{proof}
The recursion relations come from equating coefficients in the hierarchy equation~(\ref{hier}). At the leading orders, equations~(\ref{topA}, \ref{topB}) come from the off-diagonal terms at orders $\lambda^{N-1}$ and $\lambda^{N},$ respectively.
Equations~(\ref{leadingA}, \ref{leadingC}) arise from the diagonal and lower-left entries at the level of \(\lambda^{N-1}\), while equation~(\ref{leadingB}) arises from the upper-right entry at the level of \(\lambda^{N-2}\). At the lowest order, equations~(\ref{lowestA}, \ref{lowestC}) come from the $\lambda^{0}$-order terms in the diagonal entries and lower-left entries, respectively. Equation~(\ref{lowestB}) comes from the $\lambda^{-1}$-order terms in the upper-right entries.

Equations~(\ref{lowestB}, \ref{lowestC}) imply that
\begin{equation}\label{ghproduct}
\frac{\partial}{\partial t_k} \left( g_0 h_0 \right) =0,
\end{equation}
for $k=0, \ldots, N-1,$ so that $g_0 h_0$ is an invariant of all the hierarchy flows.
In particular, with sine-Gordon reality conditions, $|g_0|^2 =\frac{1}{16}$ is an invariant of all the flows. So the component $g_0$ lives on a circle of radius $1/4$ centered at the origin, and the real-valued function $\phi$ is well-defined (modulo $2 \pi$) by $g_0=-\frac{i}{4} e^{-i \phi}.$ Given a smooth solution for $g_0,$ a smooth branch of $\phi$ can be defined by choosing the initial $\phi$ in the interval $[0,2 \pi).$ Consequently, $h_0=\overline{g_0} = \frac{i}{4} e^{i \phi}.$ Similarly, with sinh-Gordon reality conditions, $g_0 >0$ and $g_0 h_0 =\frac{1}{16}$ are  invariant constraints,  and the real-valued function $\phi$ is well-defined by $g_0 = \frac{1}{4} e^{\phi}.$  Consequently, $h_0 =  \frac{1}{4} e^{-\phi}.$ Note that with initial data satisfying the sine(sinh)-Gordon reality conditions and normalization conditions, it is impossible for $g_0$ or $h_0$ to equal zero at any time, because of the invariance of the product $g_0 h_0 = \frac{1}{16} >0.$

In the case where $N=1,$ there is only one time flow, $k=0,$ which we identify with the spatial variable $t_0=x.$ Using the fact that $f_{j}=g_{j}=h_{j}=0$ for all $j<0,$  the system of equations~(\ref{leadingorder}) is simply
\begin{subequations} \label{zeroflow}
\begin{align}
f_{0x} &=  \overline{b} g_0- b h_0, \label{zeroflowA}\\[.1in]
g_{0x} &= - 2 f_0 g_0, \label{zeroflowB}\\[.1in]
h_{0x} &=  2 f_0 h_0. \label{zeroflowC}
\end{align}
\end{subequations}

Under the sine-Gordon reality conditions $b=i,$ and equation~(\ref{zeroflowB}) implies
\begin{equation}
f_0 = \frac{i}{2} \phi_x.
\end{equation}
Then equation~(\ref{zeroflowA}) gives
\begin{equation}
\frac{i}{2} \phi_{xx} = -i \left(-\frac{i}{4} e^{-i \phi} \right) - i \left(\frac{i}{4} e^{i\phi}\right),
\end{equation}
or, equivalently,
\begin{equation} \label{xequalstsG}
\phi_{xx} = \sin \phi,
\end{equation}
which is the sine-Gordon equation~(\ref{sG}) with the one-dimensional reduction $t=x.$ In particular, if $\phi(x)$ solves equation~(\ref{xequalstsG}), then 
$\Phi(X,T)=\phi(X+T)$ is a traveling-wave solution of the sine-Gordon equation in  characteristic (light-cone) coordinates,
\begin{equation}
\Phi_{XT} = \sin \Phi.
\end{equation}

Similarly, under the sinh-Gordon reality conditions $b=1,$ and equation~(\ref{zeroflowB}) implies
\begin{equation}
f_0 = - \frac{1}{2} \phi_x.
\end{equation}
Then equation~(\ref{zeroflowA}) gives
\begin{equation}
-\frac{1}{2} \phi_{xx} = \frac{1}{4} e^{\phi} - \frac{1}{4} e^{-\phi},
\end{equation}
which is equivalent to the one-dimensional reduction $t=-x$ of the sinh-Gordon equation~(\ref{shG}),
\begin{equation}\label{xequalsshG}
\phi_{xx} =- \sinh \phi.
\end{equation}
In particular, if $\phi(x)$ solves equation~(\ref{xequalsshG}), then $\Phi(X,T)=\phi (X-T)$ is a traveling-wave solution of the sinh-Gordon equation,
\begin{equation}
\Phi_{XT} = \sinh \Phi.
\end{equation}

If $N \geq 2,$ then the $k=0$ and $k=N-1$ flows together generate the sine(sinh)-Gordon equations with the identifications $t_0=x$ and $t_{N-1}=t.$ 
In this case, we use the leading-order recursion relation~(\ref{leadingA}) for $k=N-1,$ viz., 
\begin{equation} \label{Nminus1flow}
\left( f_{N-1} \right)_t  =  \overline{b}  g_{0} - b h_{0},
\end{equation}
and the lowest-order recursion relations~(\ref{lowestB}, \ref{lowestC}) for $k=0,$ viz.,
\begin{subequations}\label{lowestNminus1}
\begin{align}
\left( g_{0} \right)_{x} &= - 2 f_{N-1} g_{0},\label{lowestNminus1A}\\
\left( h_{0} \right)_{x} &= 2 f_{N-1} h_{0}.
\end{align}
\end{subequations}

Under the sine-Gordon reality conditions $b=i,$ and equation~(\ref{lowestNminus1A}) implies
\begin{equation}
f_{N-1} = \frac{i}{2} \phi_x.
\end{equation}
Then equation~(\ref{Nminus1flow}) gives
\begin{equation}
\frac{i}{2} \phi_{xt} = -i \left(-\frac{i}{4} e^{-i\phi}\right) - i \left( \frac{i}{4} e^{i\phi} \right),
\end{equation}
which is the sine-Gordon equation~(\ref{sG}),
\begin{equation}
\phi_{xt} = \sin \phi.
\end{equation}

Similarly, under the sinh-Gordon reality conditions $b=1,$ and equation~(\ref{lowestNminus1A}) implies
\begin{equation}
f_{N-1} =- \frac{1}{2} \phi_x.
\end{equation}
Then equation~(\ref{Nminus1flow}) shows that $\phi$ satisfies
\begin{equation}\label{negshG}
\phi_{xt} = -\sinh \phi.
\end{equation}
Then $\phi(x,t)$ generates a solution of the sinh-Gordon equation~(\ref{shG}) through time reversal,  $x=X, t=-T.$ In particular, if $\phi(x,t)$ is a solution of equation~(\ref{negshG}), then $\Phi (X,T)=\phi(X,-T)$ is a solution of the sinh-Gordon equation
\begin{equation}
\Phi_{XT} = \sinh \Phi.
\end{equation}

\end{proof}

The hierarchy equations naturally produce the sine-Gordon equation in
light-cone coordinates,
\begin{equation}\label{sG2}
\phi_{xt}=\sin \phi.
\end{equation}
To recover the more familiar laboratory-coordinate form, define
\begin{equation}
\Phi(X,T) = \phi\left(\frac{X+T}{2}, \frac{X-T}{2}\right).
\label{labcoords}
\end{equation}
Then equation (\ref{sG2}) becomes
\begin{equation}
\Phi_{XX}
-
\Phi_{TT}
=
\sin \Phi.
\label{sGlab}
\end{equation}
Since the finite-dimensional
hierarchy is naturally expressed in the variables $x$ and $t$, we will
continue to work in the light-cone form of equation~(\ref{sG2}).

\begin{lemma}[Invariant Polynomial]
The $2N+1$-degree polynomial $\mathscr{R}(\lambda),$ defined by
\begin{equation}
\mathscr{R}(\lambda) = -\lambda^2 \det \Psi^{(N)} (\lambda) = \lambda^2 F_{N-1}^2(\lambda) + \lambda G_{N}(\lambda) H_{N}(\lambda) = |b|^2 \lambda \sum_{j=1}^{2N} (\lambda-\lambda_{j}),
\label{Reqn}
\end{equation}
is  invariant with respect to all the time-flow variables $t_k,$ for $k=0, \ldots, N-1.$  The numbers $\lambda_j \in \mathbb{C},$ for $j=1, \ldots, 2N,$ and the fixed root at 
$\lambda =0$ are the
invariant roots of $\mathscr{R}(\lambda) = 0$ and are assumed to be distinct. 
\label{invariant}
\end{lemma}
\begin{proof}
The invariance of the determinant of $\Psi$ is a standard calculation based on Jacobi's formula for the derivative of the determinant of a matrix.  The proof can be found in~\cite{wrig 26}.
\end{proof}

Notice that the  invariant hyperelliptic spectral curve $\Gamma : w^2 = -\det \Psi^{(N)} = F_{N-1}(\lambda)^2 + G_{N} H_{N} \lambda^{-1}$ has a fixed branch point at
 $\lambda =0,$ in addition to the $2N$ roots of $\lambda F_{N-1} (\lambda)^2 + G_{N} H_{N}.$ Thus the finite-gap solution generated by $\Psi^{(N)}$ is parametrized by a hyperelliptic curve with $2N+1$ finite branch points. In addition to the branch point at infinity, this gives a genus of $N.$ In order to account for the additional branch point
at zero in the invariant hyperelliptic spectral curve, it is convenient to add the fixed root at $\lambda=0$ to the definition of the invariant polynomial $\mathscr{R}(\lambda).$

\begin{theorem}[Global Solution of the Sine-Gordon Equation] \label{globalsG}
Let $y_j=f_{j-1}, y_{N+j}= g_{j-1}, y_{2N+j} = h_{j-1},$ for $j= 1, \ldots, N,$ be a compatible local  solution of the $N$  autonomous ordinary differential equations~(\ref{Fode})
with initial data that satisfy  the invariant sine-Gordon reality conditions and normalization condition of Theorem~\ref{equationtheorem}. 
Then the solution exists for all $t_k \in \mathbb{R},$ for $k=0, \ldots, N-1,$ and is uniformly bounded on $\mathbb{R}^N.$   In particular, the density $\phi_x$ of the corresponding solution of the sine-Gordon equation~(\ref{shG}) is a uniformly bounded global function on $\mathbb{R}^N$.
\end{theorem}

\begin{proof}
The compatibility of the hierarchy flows guarantees the existence of the smooth local  solution. We will show that  the invariant polynomial \(\mathscr R\) controls the Laurent-polynomial coefficients and yields uniform bounds on all dynamical variables.

The sine-Gordon reality conditions~(\ref{realitysG}) imply that $|b|^2=1$ and, for $\lambda \in \mathbb{R},$ $F_{N-1} (\lambda) = - \overline{F_{N-1} (\lambda)}$ and $H_{N} (\lambda) =  \overline{G_{N} (\lambda)},$ so that
\begin{equation} \label{RsG}
\mathscr{R} (\lambda) = - \lambda^2 |F_{N-1}(\lambda)|^2 + \lambda |G_{N}(\lambda)|^2.
\end{equation}
Thus, for $\lambda \in \mathbb{R},$ 
\begin{equation}
\overline{\mathscr{R} (\lambda)} = \mathscr{R} (\lambda).
\end{equation}
Therefore, the coefficients of $\mathscr{R}$ are real and the roots of $\mathscr{R}(\lambda)=0$ are either real or occur in complex-conjugate pairs.

Now suppose that $\lambda < 0$ is a negative real root of $\mathscr{R}(\lambda)=0.$ Then, equation~(\ref{RsG}) implies that both $F_{N-1} (\lambda)=0$  and $G_{N}(\lambda)=0,$ which means that $\lambda$ is a double root of $\mathscr{R}(\lambda)=0,$ which is not allowed by the assumption that the roots are distinct.  Thus $\mathscr{R}(\lambda)=0$ does not have any negative real roots. All the nonzero roots are either positive or  complex-conjugate pairs. This means that there is an even number of positive roots. Therefore, 

\begin{equation}
\mathscr{R}(\lambda) = \lambda P^+(\lambda) P^- (\lambda) Q (\lambda)
\end{equation}
where $P^{+}, P^{-},$ and $Q$ are monic polynomials,  the roots of $P^{+}(\lambda)$ are all in the upper half plane, $P^-(\lambda) = \overline{P^+ (\overline{\lambda})}$, and $Q(\lambda)$ has an even number of positive roots and no other roots. Therefore,
\begin{equation}
-\lambda^2 |F_{N-1}(\lambda)|^2 + \lambda |G_{N}(\lambda)|^2 =  \lambda |P^+(\lambda)|^2  Q (\lambda).
\end{equation}
Set $\lambda = -\eta,$ where $\eta >0,$ then
\begin{equation}
\eta |F_{N-1} (-\eta)|^2 +  |G_{N} (-\eta)|^2 = |P^{+}(-\eta)|^2 Q(-\eta),
\end{equation}
and so
\begin{equation}
\label{Fbound}
\sqrt{\eta} |F_{N-1} (-\eta)| \leq \sqrt{ Q(-\eta)} |P^+ (-\eta)|
\end{equation}
and
\begin{equation}
\label{Gbound}
|G_{N} (-\eta)| \leq  \sqrt{Q(-\eta)} |P^+(-\eta)|.
\end{equation}
Note that $Q(-\eta) >0,$ since $-\eta <0$ and $Q(\lambda)$ is monic, has an even number of positive roots, and has no other roots. 
Equations~(\ref{Fbound}) and~(\ref{Gbound}) imply uniform bounds on the polynomials $|F_{N-1}(\lambda)|$ and $|G_{N}(\lambda)|$ on a closed finite interval of negative real numbers, e.g., on $-3 \leq \lambda \leq -1.$ After an affine rescaling of the interval \([-3,-1]\) onto \([-1,1]\), standard coefficient estimates for bounded polynomials (e.g., Markov inequalities~\cite{mark 90}) imply uniform bounds on all the coefficients of \(F_{N-1}(\lambda)\) and \(G_N(\lambda)\). The sine-Gordon reality conditions imply that the uniform bounds are also valid for the coefficients of \(H_N(\lambda)\).
This, in turn, implies the uniform boundedness of all the dynamical variables that appear as coefficients in the polynomials $F_{N-1}, G_{N},$ and $H_{N}.$
Since solutions to polynomial vector fields can be continued as long as the solution remains bounded, the uniform coefficient bounds allow the local solution to be continued globally.  In particular,  the density of the finite-gap solution of the  sine-Gordon equation, $\phi_x ({\bf t}) = -2 i f_{N-1} ({\bf t}),$ is uniformly bounded and exists for all times ${\bf t} \in \mathbb{R}^{N}.$ 

Notice that the uniformity of the bounds on the dynamical variables in the solution applies not only to all times but, also, to all initial conditions that share the same invariant polynomial $\mathscr{R}(\lambda).$

\end{proof}

\begin{theorem}[Global Solution of the Sinh-Gordon Equation] \label{globalshG}
Let $y_j=f_{j-1}, y_{N+j}= g_{j-1}, y_{2N+j} = h_{j-1},$ for $j= 1, \ldots, N,$ be a compatible local  solution of the $N$  autonomous ordinary differential equations~(\ref{Fode})
with initial data that satisfy  the invariant sinh-Gordon reality conditions and normalization condition of Theorem~\ref{equationtheorem}. In addition, suppose that the initial data is such that
\begin{enumerate}
\item[(a)]  all of the nonzero roots of the invariant polynomial $\mathscr{R}(\lambda)$ are distinct negative numbers  ordered so that $\lambda_{k+1} < \lambda_{k} < 0,$ for $k=1, \ldots, 2N-1,$ and
\item[(b)] the degree $N$ polynomials $G_{N}(\lambda)$ and $H_{N}(\lambda)$ each have exactly one  negative root in each of the $N$ intervals $[\lambda_{2j},\lambda_{2j-1}],$ for $j= 1, \ldots, N.$ 
\end{enumerate}
Then the solution exists for all $t_k \in \mathbb{R},$ for $k=0, \ldots, N-1,$ and is uniformly bounded on $\mathbb{R}^N.$  In particular, the density $\phi_x$ of the corresponding solution of the sinh-Gordon equation~(\ref{shG}) is a uniformly bounded global function on $\mathbb{R}^N$.
\end{theorem}

\begin{proof}
The normalization conditions for the sinh-Gordon equation in Theorem~\ref{equationtheorem} constrain the global invariant,
\begin{equation}
g_0h_0 = \prod\limits_{j=1}^{2N} \lambda_{j} = \frac{1}{16},
\end{equation}
which is compatible with the root distribution condition (a). Also, if we define $\mu_j,$ for $j=1, \ldots, N,$ to be the roots of $G_{N}(\lambda),$ then
\begin{equation}
g_0 = \prod\limits_{j=1}^{N} \left( - \mu_j \right) >0
\end{equation}
by condition (b), which is compatible with the definition of $\phi,$ viz., $g_0 = \frac{1}{4} e^{\phi}.$

The polynomials $F_{N-1}, G_{N},$ and $H_{N}$ have real coefficients and, therefore, their roots are either real or come in complex-conjugate pairs. The same symmetry is true for $\mathscr{R}(\lambda) = \lambda^2 |F_{N-1} (\lambda)|^2 + \lambda G_{N}(\lambda) H_{N} (\lambda).$ Suppose that $\lambda \in \mathbb{R}$ is a  root of $G_{N}(\lambda)$ or $H_{N}(\lambda)$, then
\begin{equation}\label{rneg}
\mathscr{R} (\lambda) = \lambda \prod\limits_{j=1}^{2N} (\lambda - \lambda_{j} ) = \lambda^2 |F_{N-1}(\lambda)|^2 \geq 0.
\end{equation}
Thus $\lambda \in \mathbb{R}$ must be a root of $\mathscr{R}(\lambda)$ or less than an even number of the $2N+1$ nonpositive roots of $\mathscr{R} (\lambda)=0,$ by condition (a).
Therefore $\lambda$ must lie in the interval $\lambda_{2j} \leq \lambda \leq \lambda_{2j-1},$ for some $j= 1, \ldots, N,$ because of the fixed root at zero. Moreover, the root 
$\lambda$ cannot leave the interval it is in and remain real since this would violate the inequality~(\ref{rneg}). However, by  condition (b), initially there is exactly one root of $G_{N}(\lambda)$ and exactly one root of $H_{N}(\lambda)$ in each interval
$[\lambda_{2j}, \lambda_{2j-1}],$ for $j=1, \ldots, N.$  Since the coefficients of both $G_{N}(\lambda)$ and $H_{N}(\lambda)$ are real, their roots are either real or occur in complex-conjugate pairs. The intervals
$[\lambda_{2j},\lambda_{2j-1}],$ for $j=1,\ldots, N,$  are disjoint, so it is impossible for two roots of $G_{N}(\lambda)$ (or $H_{N}(\lambda)$) to flow continuously, collide, and become a complex-conjugate pair, because initially only one is in each of the disjoint real intervals to which they are individually constrained. Hence, the roots remain trapped in these bounded intervals for all $t_k \in \mathbb{R},$ for $k=0, \ldots, N-1,$ for which the local solution continues to exist.  Thus the roots of $G_{N}(\lambda)$ and $H_{N}(\lambda)$ and, hence, the coefficients of $G_{N}(\lambda)$ and $H_{N}(\lambda)$ are uniformly bounded.

Moreover, since the coefficients of $G_N$ and $H_N$ are uniformly bounded as functions of ${\bf t} \in \mathbb{R}^{N},$
\begin{equation}
|\lambda|^2 |F_{N-1}(\lambda)|^2 \leq |\lambda| |G_{N}(\lambda)H_{N}(\lambda)| + |\mathscr{R}(\lambda)|
\end{equation}
implies that there exists a finite positive number $M$ such that
\begin{equation}
\max\limits_{1 \leq \lambda \leq 3} |F_{N-1} (\lambda)| \leq M,
\end{equation}
where the interval $[1,3]$ can be chosen arbitrarily as long as it excludes zero.
Using an affine transformation of the interval $[1,3]$ to $[-1,1],$ we can apply Markov's inequality  to conclude that the coefficients of $F_{N-1}(\lambda)$ are also uniformly bounded as functions of ${\bf t} \in \mathbb{R}^{N}$.

Since  the uniformly bounded coefficients of $F_{N-1}, G_{N},$ and $H_{N}$ are the components of the  compatible local solution of the commuting polynomial vector fields of equation~(\ref{Fode}), the local solution can be extended globally.  

The normalization constraint implies that $g_0$ remains positive on the global solution, so the  function $\phi$ is well-defined at every point on the global solution, and it has density $\phi_x ({\bf t}) = -2 f_{N-1} ({\bf t})$ which is uniformly bounded and exists for all times ${\bf t} \in \mathbb{R}^{N}.$

As in the sine-Gordon case, the uniformity of the bounds on the dynamical variables in the solution applies not only to all times but, also, to all initial conditions that satisfy the hypotheses of the theorem.

\end{proof}

We make the following interesting observation about a global orbit that satisfies the hypotheses of Theorem~\ref{globalshG}. If
\begin{equation}
G_{N} (\lambda) = \prod\limits_{j=1}^{N} \left( \lambda - \mu_j \right),
\end{equation}
then  $\lambda_{2j} \leq \mu_{j} \leq \lambda_{2j-1} <0,$ for all $j= 1, \ldots, N,$ and
\begin{equation}\label{explicitbounds}
0 <\prod\limits_{j=1}^{N} \left( - \lambda_{2j-1} \right) \leq g_0 = \prod\limits_{j=1}^{N} \left(- \mu_j \right) \leq \prod\limits_{j=1}^{N} \left( - \lambda_{2j} \right).
\end{equation}
Therefore, inequalities~(\ref{explicitbounds}) explicitly show that $\phi \in \mathbb{R},$ defined by $g_0 = \frac{1}{4} e^{\phi},$ is bounded and strictly bounded away from zero.

\section{Sharp Upper Bounds}

\begin{theorem}[Critical Points]\label{criticalpoints}
Let \(y_j=f_{j-1},\; y_{N+j}=g_{j-1},\; y_{2N+j}=h_{j-1}\), for \(j=1,\ldots,N\), be a compatible solution of the \(N\) autonomous ordinary differential equations~\eqref{Fode} under the sine(sinh)-Gordon reality conditions and normalization condition of Theorem~\ref{equationtheorem}.
If the point \(\mathbf t \in\mathbb R^N\) is a critical point of \(-2 i f_{N-1}:\mathbb R^N\to\mathbb R\) in the sine-Gordon case, then
\begin{equation}\label{criticalsG}
g_j(\mathbf t)=-h_j(\mathbf t),
\qquad j=0,\ldots,N-1,
\end{equation}
at \((t_0,t_1,\ldots,t_{N-1})=\mathbf t\).
If the point $\mathbf t \in \mathbb{R}^N$ is a critical point of $-2 f_{N-1}: \mathbb R^N \to \mathbb R$ in the sinh-Gordon case, then
\begin{equation}\label{criticalshG}
g_j(\mathbf t)=h_j(\mathbf t),
\qquad j=0,\ldots,N-1,
\end{equation}
at \((t_0,t_1,\ldots,t_{N-1})=\mathbf t\).
\end{theorem}

\begin{proof}
Since \(\mathbf t\) is a critical point,
\begin{equation}
\frac{\partial}{\partial t_k} f_{N-1}(\mathbf t)=0,
\qquad k=0,\ldots,N-1.
\end{equation}
Therefore equation~(\ref{leadingA}) implies
\begin{equation}
\overline{b} g_{N-k-1}(\mathbf t)= b h_{N-k-1}(\mathbf t),
\qquad k=0,\ldots,N-1.
\end{equation}
Hence
\begin{equation}\label{bb}
\overline{b} g_j(\mathbf t)= b h_j(\mathbf t),
\qquad j=0,\ldots,N-1.
\end{equation}
Under the sine-Gordon reality conditions, $b=i,$ so equation~(\ref{bb}) implies equation~(\ref{criticalsG}). Under the sinh-Gordon reality conditions, $b=1,$ so equation~(\ref{bb}) implies equation~(\ref{criticalshG}).
\end{proof}

\begin{definition} \label{squareroots}
Let
\begin{equation}
\mathscr R(\lambda)
=
|b|^2\lambda
\prod_{j=1}^{2N}
(\lambda-\lambda_j)
\end{equation}
be the invariant polynomial of Lemma~\ref{invariant}, where
$\lambda_1,\ldots,\lambda_{2N}$
denote its distinct nonzero roots. Assume that there are no negative real roots.
For each $j=1, \ldots, 2N,$ define $E_j \in \mathbb{C}$ to be the unique square root of $-\lambda_j$ with positive imaginary part, viz.,  $\lambda_j = -E_j^2$ and $\Im \left( E_j \right) > 0,$ and
define
\begin{equation}
\mathcal E^+ = \{ E_1, \ldots, E_{2N} \}.
\end{equation}
In other words,
\begin{equation}
\mathcal E^+ = \{ E : \Im \left(E \right) > 0, \mathscr{R}\left( -E^2 \right) = 0 \}.
\end{equation}
\end{definition}

\begin{theorem}[Sharp Upper Bound for the Density of the  Sine-Gordon Equation] \label{sharpsG}
Let
\[
y_j=f_{j-1}, \qquad
y_{N+j}=g_{j-1}, \qquad
y_{2N+j}=h_{j-1},
\qquad j=1,\ldots,N,
\]
be a compatible global solution of the \(N\) autonomous ordinary differential equations~\eqref{Fode}, with initial conditions satisfying the sine-Gordon reality conditions and the normalization condition of Theorem~\ref{globalsG}.
Then the density \(\phi_x\) of the corresponding solution $\phi$ of  the sine-Gordon equation~(\ref{sG}) satisfies
\begin{equation}
|\phi_x(x,t)|
\le
2 \sum_{E\in\mathcal E^+}\Im \left( E \right).
\end{equation}
Moreover, the bound is sharp, in the sense that there exists a compatible global solution whose density attains the upper bound.
\end{theorem}

\begin{proof}
Fix \(\mathscr R(\lambda)\) the invariant polynomial of the given global solution. 
Let \(\mathcal S\) denote the set of all points in the $3N$-dimensional phase space of the commuting hierarchy which (i)  satisfy the sine-Gordon reality conditions  and (ii)  satisfy the defining equation~(\ref{Reqn}) of the invariant polynomial \(\mathscr R(\lambda)\) with the sine-Gordon normalization condition. Every point of $\mathcal S$ serves as initial data for a unique solution of the commuting hierarchy which preserves the reality conditions and the invariant polynomial $\mathscr{R} (\lambda),$ therefore every point $\mathcal S$ lies on the orbit of a real  solution with invariant polynomial $\mathscr{R} (\lambda).$ Conversely, the orbit of any real solution with invariant polynomial $\mathscr{R} (\lambda)$ must lie entirely on $\mathcal S,$ because the reality conditions and $\mathscr{R}(\lambda)$ are invariants of the orbit. Hence $\mathcal S$ coincides with the union of all the real solution orbits having invariant polynomial $\mathscr{R}(\lambda).$ Equivalently, $\mathcal S$ coincides with the set of all initial data of orbits of real solutions whose invariant polynomial is $\mathscr{R}(\lambda).$
 The reality conditions define a closed subspace of the phase space, and the invariant polynomial condition is the common zero set of a finite collection of polynomials of the dynamical variables. Hence \(\mathcal S\) is closed. 

Moreover, the construction of the uniform  bounds in Theorem~\ref{globalsG}  shows that the bounds on the dynamical variables along a real orbit depend only on the invariant polynomial $\mathscr{R}(\lambda)$. Therefore, the bounds are uniform with respect to the points on the orbits of solutions having  the given invariant polynomial.   Therefore, \(\mathcal S\) is a closed  bounded subset of the finite-dimensional phase space. Therefore, \(\mathcal S\) is compact. 

Consider the projection of \(\mathcal S\) onto the \(f_{N-1}\)-coordinate of the finite-dimensional phase space. Since this projection is compact, there exists a point ${\bf s}^{\max}  \in \mathcal S$ maximizing $|f_{N-1}|.$ Note that ${\bf s}^{\max}$ is not necessarily unique, but the maximum value of $|f_{N-1}|$ for points in $\mathcal S$ is unique.  Define
${\bf y}^{\max} ({\bf t})$ to be the solution of the hierarchy flows with initial data ${\bf y}^{\max} ({\bf 0}) = {\bf s}^{\max}.$ And let $\phi_{x}^{\max} ({\bf t}) =- 2 i f_{N-1} ({\bf t})$ be the density of the corresponding global finite-gap solution of the sine-Gordon equation~(\ref{sG}) coming from the bounded global solution ${\bf y}^{\max} ({\bf t})$ of the hierarchy flows. Since every point of the orbit of ${\bf y}^{\max}( {\bf t})$ belongs to $\mathcal S,$ and the initial point was chosen to maximize $|f_{N-1}|$ on $\mathcal{S},$ it follows that
\begin{equation}
|\phi_{x}^{\max} (\bf{t} ) | \leq |\phi_{x}^{\max} (\bf{0})|,
\end{equation}
for all $\bf{t} \in \mathbb{R}^{N}.$
Therefore, \(\mathbf t=\mathbf 0\) is the location of a global maximum and a critical point of the real-valued function $\phi_{x}^{\max}({\bf t})=-2i f_{N-1}^{\max} ({\bf t}),$ with respect to all $N$ flows in the hierarchy.

Since ${\bf t} = {\bf 0}$ is a critical point of $-2i f_{N-1}^{\max} ({\bf t}),$ Theorem~\ref{criticalpoints} and $g_N =-h_N = i$ imply that
$G_N(\lambda)=- H_N(\lambda)$ at ${\bf t} = {\bf 0}.$ Therefore, at $\bf{t} = \bf{0},$ the sine-Gordon reality conditions imply that the coefficients of $G_{N} (\lambda)$ are purely imaginary, viz.,
\begin{equation}
G_N(\lambda)= -H_{N} (\lambda) = - \overline{ G ( \overline{\lambda} ) },
\end{equation}
and at ${\bf t} ={\bf 0}$ the invariant polynomial~(\ref{Reqn}) can be written as
\begin{equation}
\mathscr{R}(\lambda) = \lambda^2 F_{N-1}^2(\lambda) - \lambda G_{N}^2 (\lambda) =  \lambda \prod\limits_{j=1}^{2N} (\lambda - \lambda_j).
\end{equation}
If we define $\lambda=z^2,$ then we can factor this equation. 
By Definition~\ref{squareroots}, if $\lambda_j=-E_j^2$ is an invariant root, then the set $\mathcal E^+ = \{E_1, \ldots, E_{2N}\}$ consists of $E_j$ such that $\Im \left( E_j \right) >0.$ Recall that, in Theorem~\ref{globalsG}, we showed that all of the invariant roots $\lambda_j$ are nonnegative, so there is a unique $E_j$ with $\Im \left( E_j \right)>0$ for each $\lambda_j.$
After canceling out the root at $\lambda=0,$ we obtain
\begin{equation}\label{reducedRfocusing}
z^2F_{N-1}^2(z^2) -G_N^2(z^2) =\prod\limits_{j=1}^{2N}(z^2+E_j^2).
\end{equation}

Equation~(\ref{reducedRfocusing}) implies the factorization
\begin{equation}\label{factorfirst}
\left(z F_{N-1} (z^2) + G_{N} (z^2) \right) \left( z F_{N-1} (z^2) -  G_{N} (z^2) \right) = \prod\limits_{j=1}^{2N}(z^2+E_j^2).
\end{equation}
Suppose that $(z^2+E_j^2)$ were a factor of the first (or second) factor on the left-hand side of equation~(\ref{factorfirst}). That would imply that $z= \pm  i E_j$ were both roots of that factor and, hence, that $F(-E_j^2)=G(-E_j^2) =0.$ Equation~(\ref{reducedRfocusing}) would then force $\lambda=-E_j^2$ to be a double root of the invariant polynomial, which is not allowed.
Therefore, the factorization in equation~(\ref{factorfirst}) must have the form, using  the fact that $g_N =i,$ 
\begin{subequations}\label{factorAA}
\begin{align}
z F_{N-1} (z^2) +   G_N(z^2) &=   i\prod\limits_{j=1}^{2N} (z - \epsilon_j i E_j) \label{factorA} \\
z F_{N-1} (z^2) -   G_N (z^2) &=   -i\prod\limits_{j=1}^{2N} (z +\epsilon_j i E_j),\label{factorB}
\end{align}
\end{subequations}
where $\epsilon_j = \pm 1,$ for $j=1, \ldots, 2N.$ Notice that the factorization must have the given form, so the polynomials $F_{N-1}(z^2)$ and $G_{N}(z^2)$ are guaranteed to be well-defined by equations~(\ref{factorA}) and~(\ref{factorB}). However, we can also verify directly that the symmetric polynomials of even and odd degrees of the roots on each side of each equation give consistent expressions for $F_{N-1} (z^2)$ and $G_{N}(z^2)$.

However, the sine-Gordon reality conditions are not satisfied by every choice of the sign parameters $\epsilon_j.$ Consider the coefficient of $z^{2N-1}$ in equation~(\ref{factorA}) or~(\ref{factorB}), viz.,
\begin{equation}\label{leadingsum}
f_{N-1} = \sum\limits_{j=1}^{2N} \epsilon_j E_j,
\end{equation}
which must be purely imaginary. Since  the $\lambda_j=-E_j^2,$ for $j=1, \ldots, 2N,$ are either positive or complex-conjugate pairs, the $E_j$ are either purely imaginary or occur in pairs $E_j,-\overline{E_j},$ so  that the real part of the sum in equation~(\ref{leadingsum}) is zero if $\epsilon_j = 1,$ for all $j=1, \ldots, 2N.$ Therefore, the maximal bound of $|f_{N-1}|$ occurs when $\epsilon_{j} = 1,$ for $j=1, \ldots, 2N,$ or, equivalently, when $\epsilon_{j} = -1,$ for $j=1, \ldots, 2N,$ viz.,
\begin{equation}
|f_{N-1}^{\max}({\bf 0})|  = \sum\limits_{j=1}^{2N}  E_j = \sum\limits_{j=1}^{2N} \Im \left( E_j \right), \label{max1}
\end{equation}
since any other choice would reduce the sum by at least $2 \Im \left( E_k \right),$ for some $k=1, \ldots, 2N.$
This assumes that the distribution of the roots in the factorization corresponds to initial data in $\mathcal S.$ In other words, we assumed that the distribution of the roots in the factorization in equation~(\ref{factorAA}),   when $\epsilon_j=1$ for $j=1, \ldots, 2N,$ corresponds to initial data satisfying the sine-Gordon reality conditions.
To check that the sine-Gordon reality conditions are satisfied, solve equation~(\ref{factorAA}) when $\epsilon_{j}=1,$ for $j=1, \ldots, 2N,$ 
\begin{subequations}\label{solvedfactors}
\begin{align}
F_{N-1}(z^2) &= \frac{i}{2z} \left(\prod\limits_{j=1}^{2N} (z- i E_j) - \prod\limits_{j=1}^{2N} (z + i  E_j) \right) \\
G_{N} (z^2) &= \frac{i}{2} \left( \prod\limits_{j=1}^{2N} (z- i  E_j) + \prod\limits_{j=1}^{2N} (z + i E_j) \right).
\end{align}
\end{subequations}
Since the roots $E_j \in \mathcal E^+$ are either purely imaginary or occur in pairs $E_{j}, - \overline{E_{j}},$ label the elements of $\mathcal E^+$ as $E_{2k}=-\overline{E_{2k-1}},$ for $k=1, \ldots, M,$ and $E_{k}=-\overline{E_k},$ for $k=2M+1, \ldots, 2N,$ where $M$ is an integer between $0$ and $N.$ Then 
\begin{subequations}
\begin{align}
\prod_{j=1}^{2N} \left(z- i E_j \right) &= \prod_{k=1}^{M} \left(z^2 - i (E_{2k} - \overline{E_{2k}})z + |E_{2k}|^2 \right) \prod_{k=2M+1}^{2N} \left(z -  iE_k \right), \\
\prod_{j=1}^{2N} \left(z+ i E_j \right) &= \prod_{k=1}^{M} \left(z^2 + i(E_{2k} -\overline{E_{2k}})z + |E_{2k}|^2\right) \prod_{k=2M+1}^{2N} \left(z + i E_k \right). 
\end{align}
\end{subequations}
Therefore,
\begin{equation}
\overline{\prod_{j=1}^{2N} (\overline{z}- i E_j) }  = \prod_{j=1}^{2N} (z - i E_j), \qquad \overline{\prod_{j=1}^{2N} (\overline{z} + i E_j) }  = \prod_{j=1}^{2N} (z + iE_j).
\end{equation}
Then, equation~(\ref{solvedfactors}) implies 
\begin{equation}
F_{N-1} (z^2) = -\overline{F_{N-1} (\overline{z}^2)}, \qquad G_{N} (z^2) = -\overline{G_{N}(\overline{z}^2)}.
\end{equation}  
Since $G_N(z^2) = -H_N(z^2)$ at the critical point, we see that 
\begin{equation}
H_{N}(z^2) = \overline{G_{N}(\overline{z}^2)}.
\end{equation}
 In other words,  the sine-Gordon reality conditions are satisfied. Therefore, there is a point in $\mathcal S$ which corresponds to the maximal configuration of the roots in equation~(\ref{max1}).

Therefore, equation~(\ref{max1}) implies that the density of the finite-gap solution of the sine-Gordon equation corresponding to the maximal configuration, $\phi_{x}^{\max} ({\bf t})= 2i f_{N-1}^{\max}({\bf t}),$ satisfies
\begin{equation}
|\phi_{x}^{\max} ({\bf 0})| = 2 \sum\limits_{j=1}^{2N} \Im \left( E_j \right).
\end{equation}
Hence, for the global solution $\phi$ with the given invariant spectral polynomial,
\begin{equation}
|\phi_{x}(x,t)|
\le
2 \sum_{E \in\mathcal E^+}\Im \left(E \right),
\end{equation}
and the bound is sharp because it is attained by the density of the maximizing solution.
\end{proof}

\begin{theorem}[Sharp Upper Bound for the Density of the Sinh-Gordon Equation]\label{sharpshG}
Let
\[
y_j=f_{j-1}, \qquad
y_{N+j}=g_{j-1}, \qquad
y_{2N+j}=h_{j-1},
\qquad j=1,\ldots,N,
\]
be a compatible global solution of the \(N\) autonomous ordinary differential equations~\eqref{Fode}, with initial conditions satisfying the hypotheses of Theorem~\ref{globalshG}.
Then the density $\phi_x$ associated with the corresponding solution of the sinh-Gordon equation~(\ref{shG}) satisfies
\begin{equation}
|\phi_x(x,t)|
\le
2\sum_{j=1}^{N} \left( \beta_{2j} -  \beta_{2j-1} \right),
\end{equation}
where the roots of the invariant polynomial $\lambda_j = - \beta_j^2,$ $j=1, \ldots, 2N,$ are negative real numbers, and $\beta_j,$ $j=1, \ldots, 2N,$ are distinct positive numbers such that
\begin{equation}
0< \beta_1 < \beta_2 < \cdots < \beta_{2N-1} < \beta_{2N}.
\end{equation}  
Moreover, the bound is sharp, in the sense that there exists a solution whose density attains the upper bound.
\end{theorem}

\begin{proof}
Fix $\mathscr R(\lambda)$ the invariant polynomial of the given global solution. 
Let $\mathcal S$ denote the set of all points in the phase space of the commuting hierarchy satisfying the hypotheses of Theorem~\ref{globalshG} and having invariant polynomial $\mathscr{R}(\lambda).$ Since the reality conditions, the invariant polynomial with its the normalization condition, and the one-root-per-gap property are all preserved by the hierarchy flows, $\mathcal S$ is invariant under the commuting hierarchy. 
Consequently, $\mathcal S$ coincides with the union of all hierarchy orbits satisfying the conditions of Theorem~\ref{globalshG} and having invariant polynomial $\mathscr{R}(\lambda).$

The invariant polynomial condition with its normalization condition is the common zero set of a finite collection of polynomials of the dynamical variables. The reality conditions define a closed linear subspace of the phase space. The condition that the degree $N$ polynomials $G_{N}(\lambda)$ and $H_{N}(\lambda)$ possess exactly one root in each closed interval $[\lambda_{2j},\lambda_{2j-1}]$ is preserved under limits because roots depend continuously on coefficients and are trapped in the intervals by the reality conditions and equation~(\ref{rneg}). Thus $\mathcal S$ is a closed set under the  hypotheses of Theorem~\ref{globalshG}.

By Theorem~\ref{globalshG}, the dynamical variables admit bounds depending only on the invariant polynomial $\mathscr{R}(\lambda)$. Since $\mathscr{R}(\lambda)$ is fixed and invariant on $\mathcal S,$ these bounds are uniform over all points of $\mathcal S.$ Therefore $\mathcal S$ is bounded. 
 Therefore, \(\mathcal S\) is a closed  bounded subset of the finite-dimensional phase space. Therefore, \(\mathcal S\) is compact.

Consider the projection of \(\mathcal S\) onto the \(f_{N-1}\)-coordinate of the finite-dimensional phase space. Since this projection is compact, there exists a point ${\bf s}^{\max}  \in \mathcal S$ maximizing $|f_{N-1}|.$ Note that ${\bf s}^{\max}$ is not necessarily unique, but the maximum value of $|f_{N-1}|$ for points in $\mathcal S$ is unique.  Define
${\bf y}^{\max} ({\bf t})$ to be the solution of the hierarchy flows with initial data ${\bf y}^{\max} ({\bf 0}) = {\bf s}^{\max}.$ And let $\phi_{x}^{\max} ({\bf t}) = -2  f^{\max}_{N-1} ({\bf t})$ be the density of the sinh-Gordon equation~(\ref{shG}) coming from the bounded global solution ${\bf y}^{\max} ({\bf t})$ of the hierarchy flows. Since every point of the orbit of ${\bf y}^{\max}( {\bf t})$ belongs to $\mathcal S,$ and the initial point was chosen to maximize $|f_{N-1}|$ on $\mathcal{S},$ it follows that
\begin{equation}
|\phi_{x}^{\max} (\bf{t} ) | \leq |\phi_{x}^{\max} (\bf{0})|,
\end{equation}
for all $\bf{t} \in \mathbb{R}^{N}.$
Therefore, \(\mathbf t=\mathbf 0\) is the location of a global maximum of $|\phi_x^{\max} ({\bf t})|$ and a critical point of the real-valued function $\phi_{x}^{\max}({\bf t})=-2 f^{\max}_{N-1} ({\bf t})$ with respect to all $N$ flows in the hierarchy.

Since ${\bf t} = {\bf 0}$ is a critical point of $\phi_{x}^{\max} ({\bf t})= -2 f^{\max}_{N-1} ({\bf t}) ,$ Theorem~\ref{criticalpoints} and $g_N = h_N = 1$ imply that $G_N(\lambda)=H_N(\lambda)$ at ${\bf t} = {\bf 0}.$ Therefore, at $\bf{t} = \bf{0},$ the invariant polynomial can be written as
\begin{equation}\label{RR}
\mathscr{R}(\lambda) = \lambda^2 F_{N-1}^2(\lambda) + \lambda G_{N}^2 (\lambda) =  \lambda \prod\limits_{j=1}^{2N} (\lambda - \lambda_j).
\end{equation}
Under  our hypotheses, the invariant polynomial has negative real roots $\lambda_j = - \beta_j^2,$ for $j=1, \ldots, 2N,$ and 
\begin{equation}
0< \beta_1 < \beta_2 < \cdots < \beta_{2N-1} < \beta_{2N},
\end{equation}  
where the $\beta_j,$ $j=1, \ldots, 2N,$ satisfy the normalization condition,
\begin{equation}\label{betanormal}
\prod\limits_{j=1}^{2N} \beta_{j} = \frac{1}{4}.
\end{equation}
Setting $\lambda=z^2,$ the invariant polynomial in
equation~(\ref{RR}) reduces to 
\begin{equation}\label{reducedR}
z^2F_{N-1}^2(z^2) +G_N^2(z^2) =\prod\limits_{j=1}^{2N}(z^2+\beta_j^2).
\end{equation}
Notice that a factorization of equation~(\ref{reducedR}) containing a factor with roots distributed as
\begin{equation}\label{ff}
z F_{N-1} (z^2) + i G_{N}(z^2) = i (z^2 + \beta_k^2)\prod\limits_{j \neq k}^{2N} (z - i \epsilon_{j} \beta_j),
\end{equation}
where $\epsilon_{j} = \pm 1,$ for $j=1, \ldots, 2N,$  is not possible, because substituting $z=i \epsilon_k \beta_k$ into equation~(\ref{ff}) implies that
\begin{equation}\label{eqnep}
i \epsilon_k \beta_k F_{N-1} (-\beta_k^2) + i G_{N} (- \beta_k^2) = 0.
\end{equation}
Since $\epsilon_k = \pm 1$ and $\beta_k >0,$ equation~(\ref{eqnep}) implies that $F_{N-1} (-\beta_k^2) = G_{N}(-\beta_k^2) = 0.$ In that case, equation~(\ref{reducedR}) would imply  $\lambda= - \beta_k^2$ was a double root of the invariant polynomial $\mathscr{R}(\lambda)$, which is not possible under our assumptions.

Therefore, the factorization of equation~(\ref{reducedR}) must have the form
\begin{subequations}\label{factor00}
\begin{align}
z F_{N-1} (z^2) + i G_{N}(z^2) &= i \prod\limits_{j =1}^{2N} (z - i \epsilon_{j} \beta_j), \label{factor11} \\
z F_{N-1} (z^2) - i G_{N}(z^2) &= -i\prod\limits_{j =1}^{2N} (z + i \epsilon_{j} \beta_j). \label{factor22}
\end{align}
\end{subequations}
By considering the symmetric polynomials of even and odd degrees on both sides of equations~(\ref{factor11}) and~(\ref{factor22}), we see that the coefficients of $F_{N-1}(z^2)$ and $G_{N}(z^2)$ are real, as required by the sinh-Gordon reality conditions, for any choice of the sign parameters $\epsilon_{j} = \pm 1$.
However, we must also check that the roots of $G_{N}(\lambda)$ lie in the correct intervals.  

To determine the permissible signs $\epsilon_j = \pm 1$ in the factorization, we will use a deformation argument on the polynomials in equation~(\ref{factor00}). Allow the parameters $\beta_{2j-1}$ and $\beta_{2j}$ to continuously change in such a way that they remain distinct but each pair $\beta_{2j-1}$ and $\beta_{2j}$ coalesces in the limit.
Any continuous deformation that keeps the intervals $[-\beta_{2j}^2, -\beta_{2j-1}^2]$ distinct and bounded away from zero will work, e.g., 
\begin{equation}\label{deform}
(\beta_{2j-1}, \beta_{2j}) \mapsto \left((1-s)\beta_{2j-1} + s \beta_{2j}, \beta_{2j} \right),
\end{equation}
for $j=1, \ldots, N,$
will accomplish the desired deformation for $s \in [0,1).$ 

Assume, when $s=0,$ $G_{N}(\lambda)=H_{N}(\lambda)$ have exactly one root in each of the gaps $[-\beta_{2j}^2,-\beta_{2j-1}^2].$ For each $s \in [0,1),$ define the coefficients of the polynomials $F_{N-1}$ and $G_{N}=H_{N}$ by the factorization equations~(\ref{factor00}). We are defining and deforming the polynomials $F_{N-1}$ and $G_{N}=H_{N}$ directly through equation~(\ref{factor00}) and independently of the critical-point interpretation.  Therefore, for each $s \in [0,1),$ the coefficients of $F_{N-1}$ and $G_{N}=H_{N}$ are all real, and equation~(\ref{RR}) is true for each $s \in [0,1).$ Equation~(\ref{RR}) implies that $\lambda=\lambda_j$ can never be a root of $G_{N}(\lambda),$ since then it would also be a root of $F_{N-1}(\lambda)$ and, hence, a double root of $\mathscr{R}(\lambda),$ which is not possible in our construction. Thus, the roots of $G_{N}(\lambda)$ will vary continuously with $s \in [0,1),$ but they remain trapped in their respective gaps. 
Since the $\epsilon_{j}$ are discrete parameters, they remain constant throughout the deformation.

Passing to the limit of the deformation  at $s=1,$ the roots of $G_{N}(\lambda)$ collapse onto $\lambda_j=-\beta_{2j-1}^2 = -\beta_{2j}^2,$ for $j=1, \ldots, N.$ Moreover, these roots are distinct, viz., $\beta_{2j} \neq \beta_{2k},$ for all $j \neq k.$
Hence, in the limit at $s=1,$ equation~(\ref{reducedR}) implies that
\begin{equation}
z^2F_{N-1}^2(z^2) +\prod\limits_{j=1}^{N} (z^2+\beta_{2j}^2)^2 =\prod\limits_{j=1}^{N}(z^2+\beta_{2j}^2)^2 \Rightarrow F_{N-1}(z^2) \equiv 0.
\end{equation}
In the limit at $s=1,$ equation~(\ref{factor11}) and $F_{N-1} (E^2) \equiv 0$ imply
\begin{equation}\label{quadraticform}
  \prod\limits_{j=1}^{N} \left( z^2 + \beta_{2j}^2 \right) =  \prod\limits_{j=1}^{N} (z - i \epsilon_{2j-1} \beta_{2j})(z - i \epsilon_{2j} \beta_{2j}).
\end{equation}
Now suppose that $\epsilon_{2j-1}=\epsilon_{2j}$ in equation~(\ref{quadraticform}). Then the right-hand side of equation~(\ref{quadraticform}) has a double root at $z= i \epsilon_{2j} \beta_{2j}.$ But the factor $z^2+\beta_{2j}^2$ on the left-hand side contributes roots $z=\pm i \beta_{2j},$ which cannot both be equal to $z=i \epsilon_{2j} \beta_{2j},$ since $\beta_{2j} >0.$ Therefore, $\beta_{2j} = \beta_{2k}$ for some $ k \neq j,$ which is a contradiction, since the deformation~(\ref{deform}) avoids this possibility, even in the limit at $s=1.$ Thus,  permissible factorizations must satisfy $\epsilon_{2j-1} = - \epsilon_{2j}.$

Therefore, equating coefficients of $z^{2N-1}$ in the factorization at ${\bf t} = {\bf 0}$ given by equation~(\ref{factor11}) (or equation~(\ref{factor22})), we conclude that the maximal density is given by
\begin{equation}\label{maxshGstep}
|\phi_{x}^{\max} ({\bf 0})| =|2 f_{N-1}^{\max} ({\bf 0})| = 2\sum\limits_{j=1}^{N} \epsilon_{2j} (\beta_{2j} - \beta_{2j-1}).
\end{equation} 
for some choice of $\epsilon_{2j} = \pm 1,$ for $j=1, \ldots, N.$
Clearly the choice $\epsilon_{2j} =1,$ for $j=1, \ldots, N,$ maximizes the above sum because $\beta_{2j-1} < \beta_{2j},$ and any other choice would reduce the sum in equation~(\ref{maxshGstep}) by at least $4 (\beta_{2j} - \beta_{2j-1})$ for some $j \in \{1, \ldots, N \}.$ 

Therefore, to complete the proof,  it is sufficient to show that, if $\epsilon_{2j} = - \epsilon_{2j-1} = 1,$ for $j=1, \ldots, N,$  then the initial data for the coefficients of $F_{N-1}(\lambda)$ and  $G_{N}(\lambda)=H_{N}(\lambda)$ defined by the factorization in equation~(\ref{factor00}) produces a 
$G_{N}(\lambda)$ which has exactly one root in each of the $N$ finite gaps. Setting $\epsilon_{2j}=-\epsilon_{2j-1}=1$ in equation~(\ref{factor00}), define the polynomial
\begin{equation}
P(z) =z F_{N-1} (z^2) + i G_{N}(z^2) = i \prod\limits_{j =1}^{N} (z - i  \beta_{2j})(z+i \beta_{2j-1} ).
\end{equation}
Then
\begin{equation}\label{gp}
G_{N}(z^2) = \frac{1}{2i} ( P(z) + P(-z) )
\end{equation}
because $G_{N}(z^2)$ is completely determined by the even powers of $P(z).$ 
Then, for $\eta \in \mathbb{R},$ define
\begin{equation}
T(\eta) = i P(i \eta) = \prod\limits_{j=1}^{N} (\eta - \beta_{2j})(\eta + \beta_{2j-1}),
\end{equation} 
so substituting $z=i \eta$ in equation~(\ref{gp}) produces
\begin{equation}
G_{N}(-\eta^2) =-\frac{1}{2} (T(\eta) + T(-\eta)).
\end{equation}

Now $T(\eta)$ is an even polynomial with real coefficients and simple roots at 
\begin{equation}\label{signseq}
-\beta_{2N-1} <  \ldots < -\beta_{1} < \beta_{2} < \beta_{4} < \ldots < \beta_{2N}.
\end{equation}
In particular, the sign of $T(\eta)$ is positive as $\eta \rightarrow \pm \infty,$ and the sign alternates between the consecutive roots.
Therefore
\begin{equation}
G_{N} (-\beta_{2j}^2) G_{N} (-\beta^2_{2j-1}) =\frac{1}{4} T(-\beta_{2j}) T(\beta_{2j-1}) <0,
\end{equation}
because  $T(\beta_{2j})=T(-\beta_{2j-1})=0$ and
\begin{subequations}
\begin{align}
\operatorname{sgn} T(-\beta_{2j})
&=
(-1)^{N+j},\\
\operatorname{sgn} T(\beta_{2j-1})
&=
(-1)^{N+j+1},
\end{align}
\end{subequations}
for $j=1, \ldots, N.$
Thus,  $G_{N} (\lambda)$ changes sign at the endpoints of each interval $[-\beta^2_{2j},  -\beta^2_{2j-1}],$ for $j=1, \ldots, N,$ and the intermediate value theorem implies that $G_{N}(\lambda)$ has at least one root in each of the $N$ finite gaps. Since $G_{N}(\lambda)$ has degree $N$, we conclude that there
is exactly one root of $G_{N}(\lambda)$ in each of the required intervals with this initial data. 

Therefore, initial data specified by the factorization in equation~(\ref{factor00}) with $\epsilon_{2j}=-\epsilon_{2j-1}=1,$ $j=1, \ldots, N,$ satisfy all the hypotheses of Theorem~\ref{globalshG}. The corresponding global solution of the hierarchy defines the function $\phi$ that generates a finite-gap solution of the sinh-Gordon equation~(\ref{shG}) through time reversal. The corresponding density $\phi_x$  is the maximizing density $\phi_{x}^{\max}$ which
satisfies equation~(\ref{maxshGstep}).
Hence, for  $\phi$ corresponding to an arbitrary hierarchy solution specified by the statement of the theorem, its density $\phi_x$ must satisfy the inequality
\begin{equation}
|\phi_{x}(x,t)|
\le
2\sum\limits_{j=1}^{N} \left( \beta_{2j} - \beta_{2j-1} \right),
\end{equation}
and the bound is sharp because it is attained by the maximizing density.
\end{proof}

\section{Conclusion}

In this paper we establish a sharp upper bound for the densities of finite-gap solutions of the sine-Gordon equation. 
The proof is based entirely on commuting polynomial flows, an invariant spectral polynomial, and elementary algebraic properties of the associated dynamical variables. The critical-point analysis leads to a natural factorization of the invariant spectral polynomial in the maximizing configuration, allowing the sharp upper bound to be obtained without explicitly constructing the finite-gap solution. The resulting density formula is expressed solely in terms of the upper-half-plane square roots of the negated roots of the invariant polynomial and is shown to be sharp by constructing a maximizing configuration.

More fundamentally, the proof demonstrates that  the optimization mechanism is a property of the common finite-dimensional hierarchy shared by the modified Korteweg-de Vries equation and the sine(sinh)-Gordon equations. Although the sine-Gordon, sinh-Gordon, and modified Korteweg-de Vries equations arise through different hierarchy reductions, reality conditions, and time flows, the same critical-point mechanism produces an analogous factorization of the invariant spectral polynomial and identifies the maximizing configuration determined entirely by the spectral invariants. The results demonstrate that the optimization problem can be solved directly at the level of the finite-dimensional hierarchy, without explicit integration of the finite-gap equations. This suggests that the same optimization mechanism extends naturally to other reductions of the same hierarchy.

\appendix

\renewcommand{\appendixname}{Appendix}

\section{Relation to the AKNS Spectral Problem}
The AKNS formulation~\cite{akns 73a, akns 74, ablo 81} of the spectral problem for the sine(sinh)-Gordon equations~(\ref{sG}, \ref{shG}) studies the linear eigenvalue problem
\begin{equation} \label{spec1}
\frac{\partial}{\partial x} \left( \begin{array}{c} \psi_1 \\ \psi_2 \end{array} \right) = U \left( \begin{array}{c} \psi_1 \\ \psi_2 \end{array} \right),
\end{equation}
where
\begin{equation}
	U_{sG} = \left(\begin{array}{cc} -i E & -\frac{1}{2} \phi_x \\ \frac{1}{2} \phi_x & i E \end{array} \right), \qquad U_{shG} = \left(\begin{array}{cc} -i E & \frac{1}{2} \phi_x \\ \frac{1}{2} \phi_x & i E \end{array} \right),
\end{equation}
and the matrices $U_{sG}$ and $U_{shG}$ correspond to the sine-Gordon equation~(\ref{sG}) and the sinh-Gordon equation~(\ref{shG}), respectively.

In the finite hierarchy used in this paper, the corresponding spatial spectral problem comes from the $t_0=x$ flow, viz.,
\begin{equation} \label{spec2}
\frac{\partial}{\partial x} \left( \begin{array}{c} \psi_1 \\ \psi_2 \end{array} \right) = V^{(0)} \left(\begin{array}{c} \psi_1 \\ \psi_2 \end{array} \right),
\end{equation}
where
\begin{equation}
V^{(0)}_{sG} = \left(\begin{array}{cc} -\frac{i}{2} \phi_x & i \\ -i \lambda & \frac{i}{2} \phi_x \end{array}\right), \qquad V^{(0)}_{shG} = \left(\begin{array}{cc} \frac{1}{2} \phi_x & 1 \\ \lambda & -\frac{1}{2} \phi_x \end{array}\right),
\end{equation}
and the matrices $V^{(0)}_{sG}$ and $V^{(0)}_{shG}$ correspond to the sine-Gordon equation~(\ref{sG}) and the sinh-Gordon equation~(\ref{shG}), respectively.

The gauge equivalence established below identifies the AKNS spectral parameter $E$ with the principal-grading spectral parameter $\lambda$ through the relation $\lambda=-E^2.$

\begin{proposition}[Gauge Equivalence of the Principal-Grading and AKNS Spectral Problems]
For $E \neq 0,$ the principal-grading spectral operator $V^{(0)}_{sG}$ is gauge equivalent to the AKNS spectral operator $U_{sG},$  viz.,
\begin{equation}
V^{(0)}_{sG} = G^{-1} U_{sG} G - G^{-1} G_{x},
\end{equation}
where
\begin{equation}
G  = \left( \begin{array}{cc} -E & 1 \\ -i E & - i \end{array} \right), \qquad \lambda = -E^2.
\end{equation}
Also, for $E \neq 0,$ the principal-grading spectral operator $V^{(0)}_{shG}$ is gauge equivalent to the AKNS spectral operator $U_{shG},$ viz.,
\begin{equation}
V^{(0)}_{shG} = G^{-1} U_{shG} G - G^{-1} G_{x},
\end{equation}
where
\begin{equation}
G  = \left( \begin{array}{cc} -iE & 1 \\ -i E & - 1 \end{array} \right), \qquad \lambda = -E^2.
\end{equation} 
\end{proposition}

\begin{proof} In both cases, $G$ is independent of $x,$ so $G_x=0.$ In the sine-Gordon case, direct computation shows that
\begin{equation}
G^{-1} U_{sG} G = \left( \begin{array}{cc} -\frac{i}{2} \phi_x & i \\ i E^2 & \frac{i}{2} \phi_x \end{array} \right) = V^{(0)}_{sG},
\end{equation}
provided that $\lambda= - E^2.$

Similarly, in the sinh-Gordon case, direct computation shows that
\begin{equation}
G^{-1} U_{shG} G = \left( \begin{array}{cc} \frac{1}{2} \phi_x & 1 \\ - E^2 & -\frac{1}{2} \phi_x \end{array} \right) = V^{(0)}_{shG},
\end{equation}
provided that $\lambda= - E^2.$
\end{proof}

Consequently, the complex parameter $E$ appearing in the sharp upper bound theorem for the sine-Gordon equation is precisely the AKNS spectral parameter corresponding to the principal-grading spectral parameter $\lambda =-E^2.$ Similarly, the real parameter $\beta$ in the sharp upper bound theorem for the sinh-Gordon equation is the AKNS spectral parameter corresponding to the principal-grading spectral parameter $\lambda=-\beta^2.$

\section{Genus-One Examples and Soliton Limits}
This appendix illustrates the sharp upper bound formulas of Section 3 in the genus-one case. These calculations verify that the genus-one reduction of the finite-gap hierarchy reproduces the sharp upper bound formulas and that the classical one-soliton solutions arise in the appropriate degenerate  limit of the spectral invariants.

\subsection{Sine-Gordon Equation}
Theorem~\ref{equationtheorem} shows that, for the sine-Gordon equation with $N=1,$ the polynomial ansatz for the dynamical variables is
\begin{subequations}
\begin{align}
F_{0} (\lambda) &=  \frac{i}{2} \phi_x, \\
G_{1} (\lambda) &= i (\lambda - \mu), \\
H_{1} (\lambda) &= -i (\lambda - \overline{\mu}),
\end{align}
\end{subequations}
where a solution to the sine-Gordon equation is generated by the real-valued function $\phi$ defined by
\begin{equation}
g_0 = - i \mu = - \frac{i}{4} e^{- i \phi},
\end{equation}
and $\mu \in \mathbb{C}$ is a convenient change of variable.

The invariant polynomial in equation~(\ref{Reqn}) is
\begin{equation}\label{sGcurvesoliton}
-\frac{1}{4}\lambda^2 \phi_x^2 + \lambda (\lambda - \mu) (\lambda - \overline{\mu}) = \lambda (\lambda - \lambda_1) (\lambda - \lambda_2).
\end{equation}
Assuming $\lambda_1 \neq \lambda_2,$ substitution of $\lambda=\lambda_1$ or $\lambda=\lambda_2$ into equation~(\ref{sGcurvesoliton}) shows that neither $\lambda_1$ nor $\lambda_2$ can be a negative real number. Assuming that $\lambda_1$ and  $\lambda_2$ are distinct and nonzero, they are either two distinct positive numbers or a complex-conjugate pair. Denote the symmetric polynomials of the roots by $s_1 = \lambda_1 + \lambda_2$ and $s_2 = \lambda_1\lambda_2.$ Then the normalization condition is
\begin{equation}
s_2 = \frac{1}{16}.
\end{equation}
We can use equation~(\ref{sGcurvesoliton}) to obtain a quadratic equation for $\mu.$ The explicit solution is
\begin{equation}\label{musolved}
\mu = \frac{1}{2} \left(  s_1 -\frac{1}{4} \phi^{2}_x \pm  \sqrt{\left(\frac{1}{4}  \phi_x^2 -s_1\right)^2 - 4 s_2} \right).
\end{equation}
The reality condition requires that the solutions of equation~(\ref{musolved}) be a complex-conjugate pair, viz.,
\begin{equation}
\left(\frac{1}{4} \phi_x^2 -s_1 \right)^2  - 4 s_2 \leq 0,
\end{equation}
which implies that
\begin{equation}\label{solveinequality}
s_1 - 2 \sqrt{s_2} \leq \frac{1}{4} \phi_x^2 \leq s_1 + 2 \sqrt{s_2}.
\end{equation}
Writing $\lambda_1=-E_1^2$ and $\lambda_2 = -E_2^2,$ there are two cases: (i) $E_1= i \beta_1$ and $E_2 = i \beta_2,$ with $\beta_1, \beta_2 >0,$ and (ii) $E_1=\alpha + i \beta$ and $E_2 = -\alpha + i \beta,$ with $\beta >0.$
In the first case (two positive real roots $\lambda_1, \lambda_2$),  equation~(\ref{solveinequality}) implies 
\begin{equation}
|\phi_x| \leq 2(\beta_1 + \beta_2) = 2 \left( \Im E_1  + \Im E_2 \right).
\end{equation}
In the second case (complex-conjugate roots $\lambda_2 = \overline{\lambda_1})$, equation~(\ref{solveinequality}) implies
\begin{equation}
|\phi_x| \leq 4 \beta = 2 \left( \Im E_1 + \Im E_2 \right).
\end{equation}
In both cases, we obtain the correct sharp upper bound formula,
\begin{equation}
|\phi_x| \leq 2 \sum\limits_{E \in \mathcal E^+} \Im \left(E\right).
\end{equation}

To obtain the density of the classical kink soliton, consider the dynamical equation in equation~(\ref{zeroflowA}),
\begin{equation}
(f_0)_x = -i g_0- i h_0,
\end{equation}
which becomes
\begin{equation}\label{dynsG1}
\phi_{xx} =  2 i (\mu - \overline{\mu})=\pm 2 \sqrt{4 s_2 - \left(\frac{1}{4} \phi_x^2 - s_1 \right)^2}.
\end{equation}
The sharp upper bound is attained by an equilibrium solution for $\phi_x$ of equation~(\ref{dynsG1}), as well as a non-constant solution that oscillates between two equilibria.
In the degenerate case where $\lambda_1=\lambda_2= -E_1^2=-E_2^2 >0,$ we have $E_1 = E_2 =  i \beta,$ with $\beta>0.$  The normalization condition is $s_2=\frac{1}{16},$ so that  $\beta = \frac{1}{2}$ and $s_1= \frac{1}{2}.$ Setting $w=\phi_x,$ 
equation~(\ref{dynsG1}) becomes
\begin{equation}\label{weqn}
w_x = \pm \frac{1}{2} w \sqrt{4 - w^2}.
\end{equation}
The smooth solution of equation~(\ref{weqn}) is the  kink density,
\begin{equation}\label{kinkexample}
\phi_x = \pm 2 \sech (x- x_0),
\end{equation}
of the classical kink profile,
\begin{equation}
\phi = \pm 4 \arctan \left( e^{(x-x_0)} \right).
\end{equation}
The oscillatory orbits of $\phi_x$ degenerate into two homoclinic orbits~(\ref{kinkexample}) which connect the equilibrium at $w=0$ through the  turning points of equation~(\ref{weqn}) at $w=\pm 2.$
The kink density attains the limiting sharp upper bound of 
\begin{equation}
|\phi^{\max}_x (x_0)| = 2(\Im E_1+ \Im E_2) =4 \beta = 2.
\end{equation}

\subsection{Sinh-Gordon Equation}
Theorem~\ref{equationtheorem} shows that, for the sinh-Gordon equation with $N=1,$ the polynomial ansatz for the dynamical variables is
\begin{subequations}
\begin{align}
F_{0} (\lambda) &=   -\frac{1}{2} \phi_x, \\
G_{1} (\lambda) &=  \lambda - \mu_1, \\
H_{1} (\lambda) &=  \lambda - \mu_2,
\end{align}
\end{subequations}
where a solution to the sinh-Gordon equation is generated (under time reversal) by the real-valued function $\phi$ defined by
\begin{equation}
g_0 = -  \mu_1 =  \frac{1}{4} e^{\phi},
\end{equation}
and $\mu_1 \in \mathbb{R}$  is a convenient change of variable. Also, $h_0 = - \mu_2 = \frac{1}{4} e^{-\phi},$ for $\mu_2 \in \mathbb{R},$ by the normalization condition. 

The invariant spectral polynomial in equation~(\ref{Reqn}) is
\begin{equation}\label{shGcurve}
\frac{1}{4} \lambda^2  \phi_x^2  + \lambda (\lambda - \mu_1) (\lambda - \mu_2) = \lambda (\lambda - \lambda_1) (\lambda - \lambda_2),
\end{equation}
where we assume that $\lambda_1$ and $\lambda_2$ are distinct and nonzero.
We can use equation~(\ref{shGcurve}) to obtain a quadratic equation with roots $\mu=\mu_1, \mu_2.$ 
Denote the symmetric polynomials of the invariant roots as $s_1 = \lambda_1 + \lambda_2$ and $s_2 = \lambda_1 \lambda_2  = \frac{1}{16}$ (by the normalization condition).
The explicit solution for $\mu=\mu_1, \mu_2,$ is
\begin{equation}
\mu = \frac{1}{2} \left( \frac{1}{4} \phi_x^2 +s_1 \pm  \sqrt{\left(\frac{1}{4}\phi_x^2+s_1 \right)^2-4 s_2} \right).
\end{equation}
The reality condition is that the discriminant of the quadratic polynomial be nonnegative, so that $\mu \in \mathbb{R},$ viz.,
\begin{equation}
\left(\frac{1}{4} \phi_x^2 + s_1 \right)^2 - 4 s_2 \geq 0, 
\end{equation}
which is equivalent, for bounded solutions, to
\begin{equation}\label{eqnshGexample}
\phi_x^2 \leq -4 \left(s_1 + 2 \sqrt{s_2}\right).
\end{equation}
For a bounded solution to exist, it is necessary that $\lambda_1<0$ and $\lambda_2 <0.$ Without loss of generality, write $\lambda_1 = -\beta_1^2$ and $\lambda_2=-\beta_2^2,$ where $0 < \beta_1 < \beta_2.$ Then equation~(\ref{eqnshGexample}) implies
\begin{equation}\label{b2b1bound}
 |\phi_x| \leq 2 (\beta_2 - \beta_1),
\end{equation}
 in agreement with the sharp upper bound formula of Theorem~\ref{sharpshG}.

To obtain the density $\phi_x$ in a degenerate limit where two invariant roots coalesce, consider the dynamical equation in equation~(\ref{zeroflowA}),
\begin{equation}
(f_0)_x = g_0-  h_0,
\end{equation}
which becomes
\begin{equation}\label{dyn}
\phi_{xx} = 2( \mu_2 - \mu_1 ).
\end{equation}

For a bounded solution $\lambda_1=-\beta_1^2$and $\lambda_2=-\beta_2^2,$ with $0 < \beta_1 < \beta_2,$ so that $s_1 = -\beta_1^2-\beta_2^2$ and the normalization condition is $s_2 = \beta_1^2 \beta_2^2 = \frac{1}{16}.$ Thus, the degenerate limit $\beta_1 \rightarrow 0$ is singular, because  $\beta_2 = \frac{1}{4\beta_1} \rightarrow \infty.$ 
Equation~(\ref{dyn}) becomes
\begin{equation}\label{dynshG}
\phi_{xx} = \pm 2\sqrt{ \left( \frac{1}{4} \phi_x^2 + s_1\right)^2 - 4 s_2},
\end{equation} 
in which $s_1 = - \beta_1^2 - \frac{1}{16\beta_1^2}.$ 
The sharp upper bound is attained by equilibria solutions of equation~(\ref{dynshG}), viz.,
\begin{equation}\label{shGasyeq}
\phi_x = \pm 2 (\beta_2 - \beta_1) = \pm 2 \left(\frac{1}{4 \beta_1} - \beta_1 \right),
\end{equation}
as well as a non-constant solution that oscillates between these two turning points of equation~(\ref{dynshG}).

The leading-order behavior as $\beta_1 \rightarrow 0$ of the nonequilibrium kink-like density can be recovered by balancing terms in equation~(\ref{dynshG}) with the assumption that
\begin{equation}
\phi_x = O \left(\frac{1}{\beta_1} \right).
\end{equation}
Setting $w=\frac{1}{2} \phi_x$ and balancing the leading order terms in equation~(\ref{dynshG}) as $\beta_1 \rightarrow 0,$ we obtain the approximate equation
\begin{equation}
w_x = \frac{1}{16 \beta_1^2} - w^2,
\end{equation}
whose solution describes the leading-order behavior of the nonequilibrium density in the singular limit $\beta_1 \rightarrow 0,$ viz.,
\begin{equation}\label{phiw}
\phi_x = 2 w \sim \frac{1}{2 \beta_1} \tanh \left( \frac{1}{4 \beta_1} (x-x_0) \right).
\end{equation}
The kink density in equation~(\ref{phiw}) connects (asymptotically) the  leading-order equilibria in equation~(\ref{shGasyeq}), and its  amplitude is consistent with the sharp upper bound formula~(\ref{b2b1bound}),
\begin{equation}
|\phi_x| \leq 2 (\beta_2 - \beta_1) \sim \frac{1}{2 \beta_1}.
\end{equation}

\section*{Data Availability}

No datasets were generated or analyzed during this study.

\end{document}